\RequirePackage{ifthen}

\documentclass[reqno, 11pt, a4paper, commented, oneside]{amsart}

\usepackage{geometry}
\usepackage{graphicx}
\usepackage[ansinew]{inputenc}
\usepackage{amsmath,amsfonts,amssymb,amsthm}
\usepackage{mathtools}
\usepackage{hyperref}
\newboolean{isCommented}
\ifthenelse{\boolean{isCommented}} {
	\newcommand{\KP}[1]{\marginpar{{\small {\bf KP:} #1}}}
	\newcommand{\RS}[1]{\marginpar{{\small {\bf RS:} #1}}}
	\newcommand{\AS}[1]{\marginpar{{\small {\bf AS:} #1}}}
	\newcommand{\HT}[1]{\marginpar{{\small {\bf HT:} #1}}}
} { % else
	\newcommand{\KP}[1]{}
	\newcommand{\RS}[1]{}
	\newcommand{\AS}[1]{}
	\newcommand{\HT}[1]{}
}
\newtheorem{theorem}{Theorem}
\newtheorem{lemma}[theorem]{Lemma}

\theoremstyle{definition}

\theoremstyle{remark}

\newcommand{\ignore}[1]{} % for commenting large passages

\newcommand{\cR}{\ensuremath{\mathcal R}}

\newcommand{\E}{\ensuremath{\mathbb E}}

\newcommand{\bigO}{\ensuremath{\mathcal O}}

\long\def\symbolfootnote[#1]#2{\begingroup
\def\thefootnote{\fnsymbol{footnote}}\footnote[#1]{#2}\endgroup}
\ifthenelse{\boolean{isCommented}} {
	\geometry{%showframe,
	  hmargin={25mm, 50mm},
	  marginparwidth=40mm, 
	  vmargin={25mm, 25mm},
	  headsep=10mm,
	  headheight=5mm,
	  footskip=10mm
	}
} { % else
	\geometry{%showframe,
	  hmargin={25mm, 25mm}, 
	  vmargin={25mm, 25mm},
	  headsep=10mm,
	  headheight=5mm,
	  footskip=10mm
	}
}
	
\begin{document}

%%%%%%%%%%%%%%%%%%%%%%%%%%%%%%%%%%%%%%%%%%%%%%%%%%%%%%%%%%%%%%%%%%%%%%
% Title Matters
%%%%%%%%%%%%%%%%%%%%%%%%%%%%%%%%%%%%%%%%%%%%%%%%%%%%%%%%%%%%%%%%%%%%%%

\begin{center}

\LARGE Explosive Percolation in Erd\H os-R\'enyi-Like Random Graph Processes
\vspace{2mm}

\Large Konstantinos Panagiotou \quad Reto Spöhel\symbolfootnote[1]{The author was supported by a fellowship of the Swiss National Science Foundation}
\vspace{2mm}

\large
  Max Planck Institute for Informatics \\
  66123 Saarbr\" ucken, Germany \\
  {\small\{{\tt kpanagio|rspoehel}\}{\tt @mpi-inf.mpg.de}}
  
\vspace{5mm}

\Large Angelika Steger \quad Henning Thomas\symbolfootnote[2]{The author was supported by the Swiss National Science Foundation, grant 200021-120284.}
\vspace{2mm}

\large
  Institute of Theoretical Computer Science \\
  ETH Zurich, 8092 Zurich, Switzerland \\
  {\small\{{\tt steger|hthomas}\}{\tt @inf.ethz.ch}}
\vspace{5mm}

\small

\begin{minipage}{0.8\linewidth}
\textsc{Abstract.}
The evolution of the largest component has been studied intensely in a variety of random graph processes, starting in 1960 with the Erd\H os-R\'enyi process. It is well known that this process undergoes a phase transition at $n/2$ edges when, asymptotically almost surely, a linear-sized component appears. Moreover, this phase transition is continuous, i.e., in the limit the function $f(c)$ denoting the fraction of vertices in the largest component in the process after $cn$ edge insertions is continuous. A variation of the Erd\H os-R\'enyi process are the so-called Achlioptas processes in which in every step a random pair of edges is drawn, and a fixed edge-selection rule selects one of them to be included in the graph while the other is put back. Recently, Achlioptas, D'Souza and Spencer \cite{MR2502843} gave strong numerical evidence that a variety of edge-selection rules exhibit a discontinuous phase transition. However, Riordan and Warnke \cite{Warnke2011} very recently showed that all Achlioptas processes have a continuous phase transition. In this work we prove discontinuous phase transitions for a class of Erd\H os-R\'enyi-like processes in which in every step we connect two vertices, one chosen randomly from all vertices, and one chosen randomly from a restricted set of vertices.
\end{minipage}

\end{center} 

\vspace{5mm}

%%%%%%%%%%%%%%%%%%%%%%%%%%%%%%%%%%%%%%%%%%%%%%%%%%%%%%%%%%%%%%%%%%%%%%
% Introduction
%%%%%%%%%%%%%%%%%%%%%%%%%%%%%%%%%%%%%%%%%%%%%%%%%%%%%%%%%%%%%%%%%%%%%%

\section{Introduction}
In their seminal paper \cite{MR0125031} from 1960 Erd\H os and R\'enyi analyze the size of the largest component in the prominent random graph model $G_{n,m}$, a graph drawn uniformly at random from all graphs on $n$ vertices with $m$ edges. For any graph $G$ let $L_1(G)$ denote the size of its largest component. We say that an event occurs \emph{asymptotically almost surely} (a.a.s.) if it occurs with probability $1-o(1)$ as $n$ tends to infinity.

\begin{theorem}[\cite{MR0125031}] \label{thm:classical}
For any constant $c > 0$ the following holds.
\begin{itemize}
	\item If $c < 0.5$, then a.a.s.\ $L_1(G_{n, \lceil cn \rceil}) = \mathcal O(\log n)$.
	\item If $c > 0.5$, then a.a.s.\ $L_1(G_{n, \lceil cn \rceil}) = \Omega(n)$, and all other components have $\mathcal O(\log n)$ vertices.
\end{itemize}
\end{theorem}

This result can be seen from a random graph \emph{process} perspective. Starting with the empty graph on the vertex set $[n] := \{1, 2, \dots, n\}$, we add a single edge chosen uniformly at random from all non-edges in every step. It is not hard to see that the graph after inserting $m$ edges is distributed as $G_{n,m}$. In this context Theorem~\ref{thm:classical} states that, asymptotically, a linear-sized component (a so called `giant') appears around the time when we have inserted $n/2$ edges. We say that the Erd\H os-R\'enyi process has a \emph{phase transition} at $n/2$.

This phase transition has been studied in great detail (for a survey see Chapter 3 in \cite{purple}). It is known that $L_1(G_{n,cn})$ a.a.s.\ satisfies $L_1(G_{n,cn}) = (1+o(1)) f(c) n$ for some continuous function $f(c)$ with $f(c) = 0$ for every $c < 0.5$ and $\lim_{c \rightarrow 0.5^+} f(c) = 0$. Thus, the phase transition in the Erd\H os-R\'enyi process is continuous\footnote{In the literature such a phase transition is also called \emph{second order} while discontinuous phase transitions are called \emph{first order}.} (see Figure~\ref{fig:plots}).

A variant of the Erd\H os-R\'enyi process which gained much attention over the last decade, mostly concerning the question whether one can delay or accelerate the appearance of the giant component \cite{MR1848028,MR2375718}, is given by the so-called \emph{Achlioptas processes}. In this variant one starts with the empty graph on the vertex set $[n]$, and in every step gets presented a pair of edges chosen uniformly at random from all pairs of non-edges in the current graph. A fixed edge-selection rule then selects exactly one of them to be inserted into the graph while the other is put back into the pool of non-edges.

\begin{figure}[tb]
\centering
%\vspace{154pt}
\begin{tabular}{cc}
\includegraphics[width=190pt]{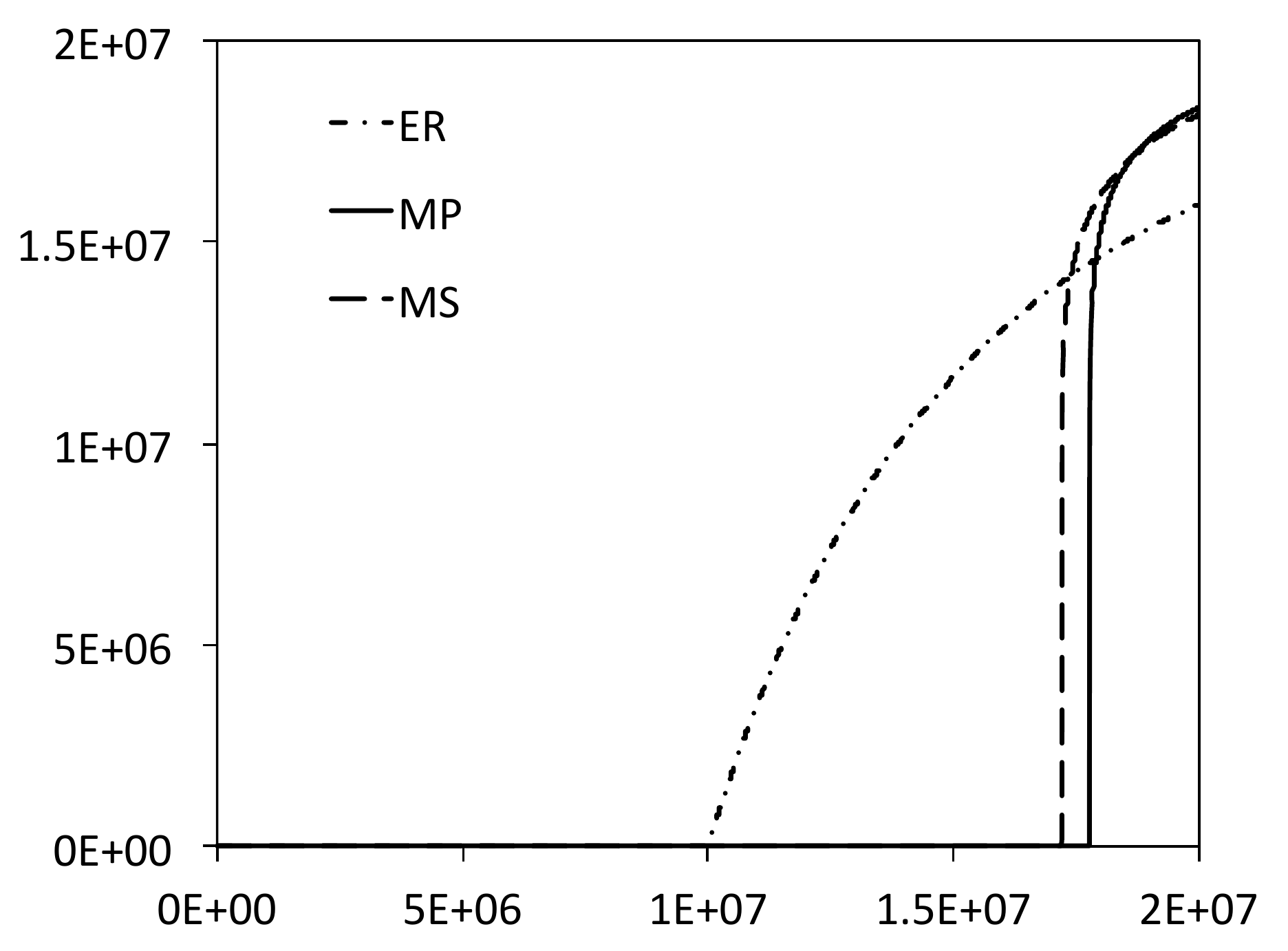} &
\includegraphics[width=190pt]{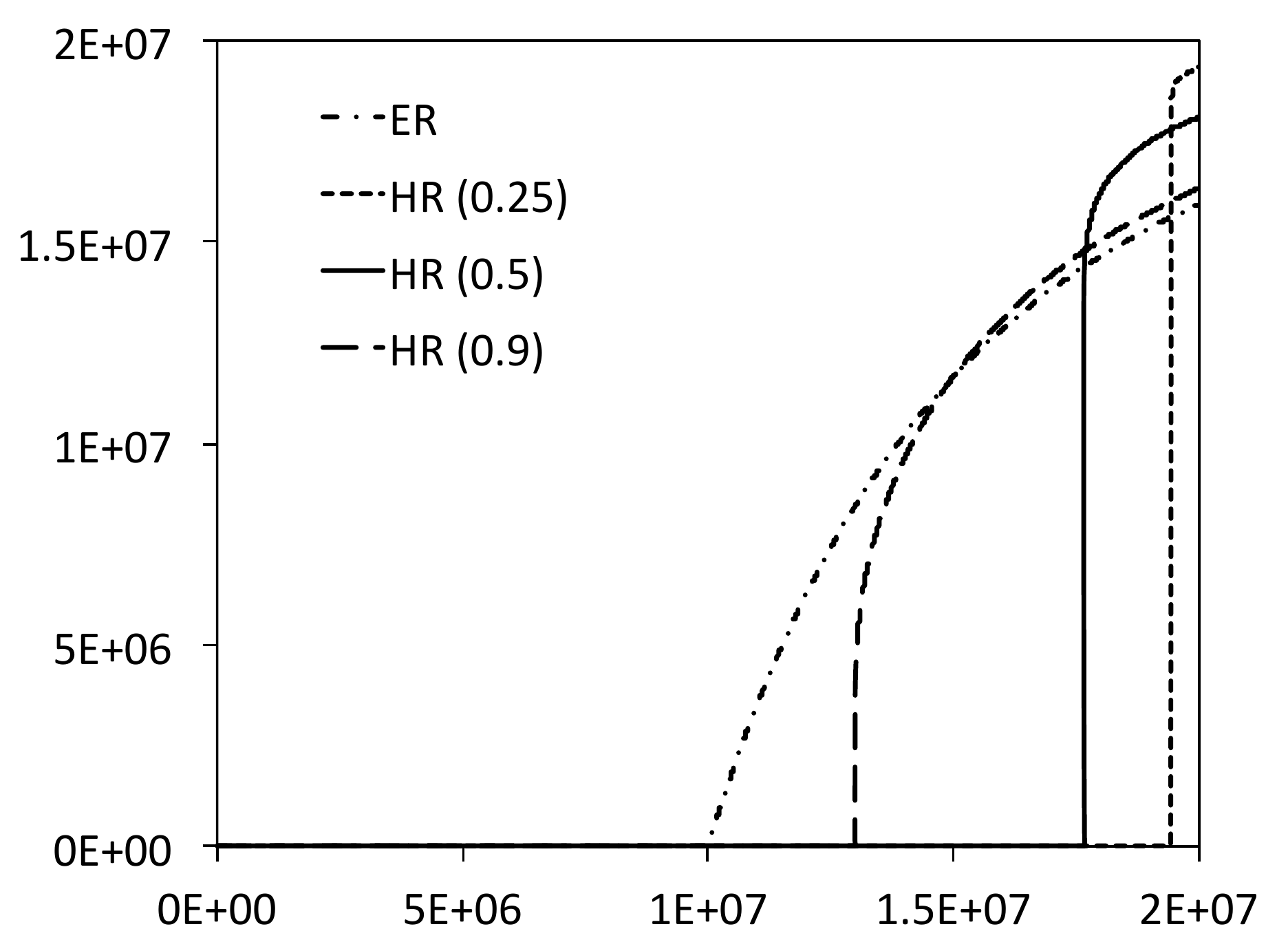}
\end{tabular}
\caption{Evolution of the largest component over the first $n$ edges of \textbf{E}rd\H os-\textbf{R}\'enyi process, \textbf{m}in-\textbf{p}roduct and \textbf{m}in-\textbf{s}um rule and the \textbf{h}alf-\textbf{r}estricted process with parameters 0.25, 0.5 and 0.9 on 20 million vertices}
\label{fig:plots}
\end{figure}

Recently, Achlioptas, D'Souza and Spencer \cite{MR2502843} provided strong numerical evidence that the min-product (select the edge that minimizes the product of the component sizes of the endpoints) and min-sum rule (select the edge that minimizes the respective sum) exhibit discontinuous transitions, in contrast to a variety of closely related edge-selection rules, in particular the ones analyzed in \cite{MR2375718}. The authors conjectured that a.a.s.\ the number of edge insertions between the appearance of a component of size $\sqrt{n}$ and one of size $n/2$ is at most $2 n^{2/3} = o(n)$, that is, at the phase transition a constant fraction of the vertices is accumulated into a single giant component within a sublinear number of steps (see Figure~\ref{fig:plots}). This phenomenon is also called \emph{explosive percolation} and of great interest for many physicists. Thus, a series of papers has been devoted to understanding this phase transition (see e.g. \cite{PhysRevE.81.030103,PhysRevLett.103.135702,PhysRevLett.104.195702,PhysRevE.81.036110}), most of the arguments not being rigorous but supported by computer simulations.

Countering the numerical evidence it was claimed in \cite{PhysRevLett.105.255701} that the transition is actually continuous. Recently, Riordan and Warnke in \cite{Warnke2011} indeed confirmed this claim with a rigorous proof. In fact, their argument shows continuous phase transitions for an even larger class of processes. A random graph process $(G(T))_{T \ge 0}$ is called an $\ell$-vertex rule if $G(0)$ is empty, and $G(T)$ is obtained from $G(T-1)$ by drawing a set $V_T$ of $\ell$ vertices uniformly at random and adding only edges within $V_T$, where at least one edge has to be added if all vertices in $V_T$ are in different components. Observe that every Achlioptas process is a $4$-vertex rule.

%%%%%%%%%%%%%%%%%%%%%%%%%%%%%%%%%%%%%%%%%%%%%%%%%%%%%%%%%%%%%%%%%%%%%%
% Our Results
%%%%%%%%%%%%%%%%%%%%%%%%%%%%%%%%%%%%%%%%%%%%%%%%%%%%%%%%%%%%%%%%%%%%%%

\section{Our Results}

In this work we prove that a variant of the Erd\H os-R\'enyi process which we call \emph{half-restricted process} exhibits a \emph{discontinuous} phase transition. Intuitively speaking, a discontinuous transition can only occur if one avoids connecting two components that are already large. The idea to avoid this in an Achlioptas process is to present two edges and choose the one that connects smaller components, which essentially is what the min-product and min-sum rule do. However, we pursue a different approach here. We connect two vertices in every step, but we restrict one of them to be within the smaller components.

For every $0 < \beta \le 1$ and every labeled graph $G$ with $n := |V(G)|$ we define the restricted vertex set $R_\beta(G)$ to be the $\lfloor \beta n \rfloor$ vertices in smallest components. Precisely, let $v_1, v_2, \dots, v_n$ be the vertices of $G$ sorted ascending in the sizes of the components they are contained in, where vertices with the same component size are sorted lexicographically. Then $R_\beta(G) := \{v_1, v_2, \dots, v_{\lfloor \beta n \rfloor}\}$.

The half-restricted process has a parameter $0 < \beta \le 1$ and starts with the empty graph $\cR_\beta(0) = \cR_{n,\beta}(0)$ on the vertex set $[n]$. In every step $T \ge 1$ we draw one \emph{unrestricted} vertex $v_1 \in [n]$ uniformly at random and, independently, one \emph{restricted} vertex $v_2 \in R_\beta(\cR_\beta(T-1))$ uniformly at random. We obtain $\cR_\beta(T)$ by inserting an edge between $v_1$ and $v_2$ if the edge is not already present (in which case we do nothing). Note that for $\beta < 1$ the half-restricted process is not an $\ell$-vertex rule.

For a half-restricted process let $\alpha_T$ be the random variable that denotes the maximum size of all components that the restricted vertex can be drawn from in step $T$. Clearly, $\alpha_T$ is increasing in $T$. For every positive integer $k$ we denote the random variable for the last step when $\alpha_T$ is still below $k$ by
$$
T_k := \max \{ T : \alpha_T < k \} \enspace.
$$

Our main result is that for any parameter $d < 1$ the half-restricted process exhibits explosive percolation. Thus, even though Figure~\ref{fig:plots} suggests that the phase transitions of the min-product (or min-sum) rule and the half-restricted process behave similarly, their mathematical structure is fundamentally different.
\begin{theorem}[Main Result] \label{thm:hr_main}
Let $0 < \beta < 1$. For every $K = K(n)$ with $\ln (n)^{1.02} \le K \le n$, every $C = C(n)$ with $1 \ll C \ll \ln(K)$, and every $\varepsilon > 0$ we have a.a.s.\ that
\begin{itemize}
\item[(i)] $L_1(\cR_\beta( T_C )) \le K$, and
\item[(ii)] $L_1(\cR_\beta( T_C + n/\ln(C))) \ge (1 - \varepsilon) (1-\beta) n$.
\end{itemize}
\HT{Replace $\ln(C)$ in $(ii)$ by a arbitrary $D \ll \sqrt{C}$}
\end{theorem}

Note that for $K = \ln(n)^2$ and $C = \ln \ln \ln n$ the theorem shows that, a.a.s., the number of steps from the first appearance of a component of size $\ln(n)^2$ to the appearance of a component of size $((1-\beta)/2) \cdot n$ is $\bigO(n/\ln \ln \ln \ln n) = o(n)$.

%%%%%%%%%%%%%%%%%%%%%%%%%%%%%%%%%%%%%%%%%%%%%%%%%%%%%%%%%%%%%%%%%%%%%%
% Proof of the Main Result
%%%%%%%%%%%%%%%%%%%%%%%%%%%%%%%%%%%%%%%%%%%%%%%%%%%%%%%%%%%%%%%%%%%%%%

\section{Proof of the Main Result}

In our proof we need a rather technical lemma introduced in the following.

Let $N$ be a positive integer, and for every $0 \le i < N$ let $X_i \sim \mathcal Geom(\frac{N-i}{N})$ be a geometrically distributed random variable and set $X(a,b) := \sum_{i=a}^{b} X_i$ for every $0 \le a < b < N$. All subsequent statements and arguments are about sums of geometrically distributed random variables, but we will use that they have the following combinatorial interpretation to a coupon collector scenario with $N$ coupons. (In a coupon collector scenario, we have a number of different coupons and repeatedly draw one uniformly at random with replacement. We are interested in how often we have to draw until we have seen every coupon.) Observe that $X_i$ is distributed like the number of coupons we need to draw while holding exactly $i$ different coupons (waiting for the $(i+1)st$), and thus $X(a,b)$ can be seen as the number of coupons we draw while holding between $a$ and $b$ different coupons or, equivalently, to obtain $b-a+1$ coupons of a fixed subset of $N-a$ coupons.

Note that
\begin{equation} \label{eq:partial_cc}
\E[X(a,b)] = \sum_{i=a}^{b} \frac{N}{N-i} = N \sum_{i=N-b}^{N-a} \frac1i = N (H_{N-a} - H_{N-b-1}) \enspace,
\end{equation}
where $H_n = \sum_{i=1}^{n} 1/i$ denotes the $n$-th harmonic number for every $n \ge 1$.
The following lemma concerns the situation when $a = N - \omega(1)$ and $b = N-2$. In terms of coupon collector, we want to collect all but $1$ coupon of a fixed subset of $\omega(1)$ coupons. We show an exponential bound on the probability that the number of coupons we need to draw for this is much lower than its expectation.

\begin{lemma} \label{lem:partial_cc}
Let $k = k(N) = \omega(1)$ and $s = s(N) \ll N \ln(k)$. % = \Theta(\E[ X(N-k,N-2) ])$.
Then for $N$ large enough
$$
\Pr[ X(N-k,N-2) \le s ] \le \textup{e}^{-k^{0.99}} \enspace.
$$
\end{lemma}

\begin{proof}
First note that by \eqref{eq:partial_cc} we have
\begin{align*}
\E[ X(N-k,N-2) ] = N (H_{k} - H_1) \stackrel{k=\omega(1)}= (1 + o(1)) N \ln (k) \enspace.
\end{align*}
Thus, $s \ll \E[X(N-k,N-2)]$. Consider $N$ different coupons and let $\{c_1, c_2, \dots, c_k\}$ be a fixed subset of them. We now draw $s$ coupons with replacement and call the resulting set $F$. For every $1 \le j \le k$ let $Y_j$ be an indicator variable for the event that coupon $c_j$ is not drawn within these $s$ trials, i.e., $c_j \notin F$. By the comments preceding this lemma, the probability of the event $X(N-k,N-2) \le s$ equals the probability that at most $1$ of these coupons is not drawn within the $s$ trials. Hence, for $Y := \sum_{j=1}^{k} Y_j$ we observe that
\begin{align} \label{eq:cc_chernoff_relation}
\Pr[ X(N-k, N-2) \le s] = \Pr [ Y \le 1] \enspace.
\end{align}
Using the identity $1-x \ge \textup{e}^{-2x}$ which holds for all $0 \le x \le 1/2$ we have for every $1 \le j \le k$ that
$$
\E[ Y_j ] = \Pr[ Y_j = 1] = \Big( 1 - \frac1N \Big)^s \ge \textup{e}^{-\frac{2s}N} \enspace,
$$
and thus, for $N$ large enough,
\begin{equation} \label{eq:EY}
\E[ Y ] \ge k \textup{e}^{- \frac{2s}N} = k^{1 - \frac{s}{N \ln(k)}}  \stackrel{s \ll N \ln(k)}\ge 8 k^{0.99} \enspace.
\end{equation}
One can check that the random variables $Y_1, Y_2, \dots, Y_k$ are negatively associated (see e.g.\ Chapter 3 in \cite{dubhashi2009concentration}). Hence, we can apply Chernoff bounds to $Y$ and obtain for $N$ large enough that
$$
\Pr[ Y \le 1 ] \stackrel{\eqref{eq:EY}}\le \Pr \left[ Y \le \Big(1 - \frac12 \Big) \E[Y] \right] \le \textup{e}^{-\E[Y]/8} \stackrel{\eqref{eq:EY}}\le \textup{e}^{-k^{0.99}} \enspace,
$$
which together with \eqref{eq:cc_chernoff_relation} finishes the proof.
\end{proof}

We now turn to the proof of Theorem~\ref{thm:hr_main}.

\begin{proof}[Proof (of Theorem~\ref{thm:hr_main})]
We fix $\beta < 1, K = K(n)$ with $\ln(n)^{1.02} \le K \le n$, $C = C(n)$ with $1 \ll C \ll \ln(K)$, and $\varepsilon > 0$. To simplify notation we write $\cR(T)$ instead of $\cR_\beta(T)$.

We first address $(i)$. We need to show that at the step when the restricted vertex can for the first time be in a component of size $C$, there is a.a.s.\ no component of size larger than $K$. The main idea is that a large component needs to be drawn by the unrestricted vertex so often that this is unlikely to happen within $T_C$ steps.

Note that up to step $T_C$ two components of size at least $C$ can never be merged by an edge since the restricted vertex in every step is drawn from vertices in components of size less than $C$. Hence, we can easily keep track of these components. We call a component $A$ in $\cR(T)$ a \emph{chunk} if it has size at least $C$. Let $A_1, A_2, \dots$ denote all chunks in order of appearance during the process. By the pigeonhole principle, there can be at most $n/C$ chunks. For every $1 \le i \le n/C$ we denote by $\mathcal E_i$ the event that chunk $A_i$ has size larger than $K$ in $\cR(T_C)$. We will show that $\Pr[ \mathcal E_i ] \le 1/n$ for every $1 \le i \le n$. By applying the union bound this implies
\begin{equation} \label{eq:unionbound_Ei}
\Pr[ \cR(T_C) \text{ contains a comp. of size} > K ] = \Pr \left[ \bigcup_{i = 1}^{n / C} \mathcal E_i \right] \le \frac nC \cdot \frac1n \stackrel{C \gg 1}= o(1) \enspace.
\end{equation}
It remains to bound $\Pr[\mathcal E_i]$ for every $1 \le i \le n/C$. Let $1 \le i \le n/C$ be fixed for the remainder of the proof. Since for every $T \le T_C$ the restricted vertex is drawn from vertices in components of size smaller than $C$, the chunk $A_i$ can grow by at most $C$ in every step. Hence, the chunk has size at most $(j+1) \cdot C$ before a vertex from the chunk is drawn for the $j$th time. Since moreover, a chunk has size at most $2C$ when it appears, a vertex from the chunk needs to be drawn in at least $K/C - 1$ steps after its appearance for $\mathcal E_i$ to happen. Let $X_2, X_4, \dots, X_{K/C}$ denote the number of steps between steps in which we draw a vertex from $A_i$. That is, $X_2$ is the number of steps from the appearance of the chunk until a vertex from $A_i$ is drawn for the first time, $X_3$ is the time from that step until a vertex from $A_i$ is drawn for the second time, and so on. Furthermore, let $X := \sum_{j=2}^{K/C} X_j$. Then,
\begin{equation} \label{eq:relation_Ei_X}
\Pr[ \mathcal E_i] \le \Pr[ X \le T_C] \enspace,
\end{equation}
and it thus suffices to show that $\Pr[ X \le T_C ] \le 1/n$. Clearly, only the unrestricted vertex in every step can be in $A_i$, and thus, $X_j$ is stochastically dominated by a geometrically distributed random variable with parameter $jC/n$. Hence, let $Y_j$ be a random variable with $Y_j \sim \mathcal Geom(jC/n)$ and set $Y := \sum_{j=2}^{K/C} Y_j$. Then
\begin{equation} \label{eq:relation_X_Y}
\Pr[ X \le T_C ] \le \Pr[ Y \le T_C ] \enspace.
\end{equation}
This is exactly the setup of Lemma~\ref{lem:partial_cc} with $N = n/C$, $k = K/C$ and $s = T_C$. Concerning the prerequisites of the lemma, we have $k = \omega(1)$ since $C \ll \ln(K) \ll K$. Furthermore, it is not hard to show that $T_C \le 4n$ (see e.g. Lemma~3 and the remark following the proof in \cite{Warnke2011}) and thus
$$
N \ln(k) = \frac nC \cdot \ln \Big( \frac KC \Big) \stackrel{C \ll \ln(K)} = (1-o(1)) \frac {n \ln(K)}{C} \stackrel{C \ll \ln(K)} \gg 4n \ge T_C = s \enspace.
$$
Hence, we can apply Lemma~\ref{lem:partial_cc}, which gives us for $n$ large enough that
$$
\Pr[ Y \le T_C ] \le \textup{e}^{-(K/C)^{0.99}} \enspace.
$$
Since $K \ge \ln(n)^{1.02}$ and $C \ll \ln(K)$ we have
$$
\Pr[ Y \le T_C ] \le \textup{e}^{-\ln(n)} = \frac1n \enspace.
$$
Using \eqref{eq:relation_X_Y}, \eqref{eq:relation_Ei_X} and \eqref{eq:unionbound_Ei} this settles the proof of $(i)$, and it remains to prove $(ii)$, i.e., we have to show that $\cR(T_C + n/\ln(C))$ contains a component of size $(1-\varepsilon)(1-\beta)n$ with high probability.

We set $a := n/(2 \ln(C))$ and split the proof into two parts. In the first part (the first $a$ additional steps after $T_C$) we collect a suitable amount of vertices in components of size at least $C$, and in the second part (the remaining $a$ steps) we actually build a giant component on these vertices.

Consider the steps $T_C + 1$ to $T_C + a$ in the graph process. Let $U(T)$ denote set of vertices in components of size at least $C$ in $\cR(T)$. We show a lower bound on the size of $U(T_C+a)$ that holds with high probability. Note that by definition of $T_C$ we have $|U(T_C+1)| \ge (1 - \beta)n$. If $|U(T_C+a)| \ge (1-\beta/2)n$, we have a bound that is by far good enough for our purposes, hence we restrict ourselves to the case $|U(T_C + i)| \le (1-\beta/2)n$ for every $1 \le i \le a$.

For every $1 \le i \le a$ let $X_i := |U(T_C + i)| - |U(T_C+i-1)|$ denote the number of vertices added to the components of size at least $C$ in the $i$th additional step. Furthermore, let $X := \sum_{i=1}^{a} X_i$. We lower bound the probability that $X_i$ contributes at least 1 vertex. Clearly, this happens if the unrestricted vertex is drawn from components of size at least $C$, which by definition of $T_C$ happens with probability at least $1-\beta$, and if the restricted vertex is drawn from components of size smaller than $C$, which happens with probability at least $1/2$ since we assume $|U(T_C+i)| \le (1-\beta/2)n$. Hence, for every $1 \le i \le a$ we have
$$
\Pr[ X_i \ge 1 ] \ge \frac{1-\beta}2 \enspace.
$$
Hence, $\E[X] \ge a(1-\beta)/2 = (1-\beta)n/(4\ln(C))$. In particular, $X$ stochastically dominates a sum of $a$ independent Bernoulli random variables with parameter $(1-\beta)/2$, such that we can use Chernoff bounds to obtain
$$
\Pr \left[ X < \frac12 \cdot \frac{(1-\beta) n}{4 \ln(C)} \right] \le \textup{e}^{-\Theta(\frac{n}{\ln(C)})} = o(1) \enspace.
$$
Hence, we can in the following condition on the event that
\begin{align} \label{eq:Cregime_lb}
|U(T_C+a)| \ge (1-\beta)n + \frac{1-\beta}{8 \ln(C)} n \enspace.
\end{align}
We now look at the second half of additional steps, i.e., steps $T_C+a+1$ to $T_C+ 2a$, and show that a.a.s. in these steps a sufficiently large component is created within $U := U(T_C+a)$. Note that $U$ is a fixed set of vertices which does not change from step to step.

Assume that $\cR(T_C + 2a)$ has no component of size $(1-\varepsilon)(1-\beta)n$. We now show that a.a.s.\ this assumption leads to a contradiction.

We call a step \emph{successful} if it connects two components in $U$. Since every component in $U$ has size at least $C$ we have that
\begin{align} \label{eq:num_suc_steps}
\frac{(1-\beta)n}{C}
\end{align}
successful steps will connect all components in $U$ such that $U$ forms one giant component of size at least $(1-\beta)n$. We now compute the probability to have a successful step if we do not have a component of size $(1 - \varepsilon) (1-\beta) n$. For a successful step, the restricted vertex needs to be in $U$ which happens with probability at least
$$
\frac{|U| - (1-\beta)n}{\beta n} \stackrel{\eqref{eq:Cregime_lb}}\ge \frac{1-\beta}{8 \beta \ln(C)} \enspace,
$$
and the unrestricted vertex needs to be drawn from a different component in $U$, which happens with probability at least
$$
\frac{|U| - (1 - \varepsilon)(1-\beta)n}{n} \stackrel{\eqref{eq:Cregime_lb}}\ge \varepsilon(1-\beta) \enspace.
$$
Hence, the probability to have a successful step is at least
$$
p := \frac{\varepsilon(1-\beta)^2}{8 \beta \ln(C)} \enspace.
$$
Thus, the number $T$ of successful steps stochastically dominates a sum of $a$ independent Bernoulli distributed random variables with parameter $p$, and we obtain that
\begin{align} \label{eq:ET}
\E[T] \ge ap = \frac{n}{2 \ln(C)} \cdot \frac{\varepsilon(1-\beta)^2}{8 \beta \ln(C)} = \frac{\varepsilon(1-\beta)^2}{16 \beta} \cdot \frac{n}{\ln(C)^2} \enspace,
\end{align}
and by Chernoff bounds
\begin{align*}
\Pr \Big[ T \le \frac{(1-\beta)n}{C} \Big] &= \Pr \left[ T \le \left( 1 - \Big(1 - \frac{16 \beta \ln(C)^2}{(1-\beta) \varepsilon C} \Big) \right) \frac{\varepsilon(1-\beta)^2}{16 \beta} \cdot \frac{n}{\ln(C)^2} \right] \\
& \stackrel{\eqref{eq:ET}}\le \Pr \left[ T \le \left(1 - \Big(1 - \frac{16 \beta \ln(C)^2}{(1-\beta) \varepsilon C} \Big) \right) \E[T] \right] \\
& \stackrel{\ln(C)^2/C = o(1)}\le \Pr \Big[ T \le \frac12 \E[T] \Big] \\
& \le \textup{e}^{-\E[T] / 8} \\
& \stackrel{\eqref{eq:ET}}= \textup{e}^{-\Omega( \frac{n}{\ln(C)^2})} = o(1) \enspace.
\end{align*}
Hence, we a.a.s.\ have at least $(1-\beta)n/C$ successful steps, which together with the considerations around \eqref{eq:num_suc_steps} contradicts the assumption to have no component of size $(1-\varepsilon)(1-\beta)n$.
\end{proof}

%%%%%%%%%%%%%%%%%%%%%%%%%%%%%%%%%%%%%%%%%%%%%%%%%%%%%%%%%%%%%%%%%%%%%%
% Bibliography
%%%%%%%%%%%%%%%%%%%%%%%%%%%%%%%%%%%%%%%%%%%%%%%%%%%%%%%%%%%%%%%%%%%%%%

\bibliography{refs}
\bibliographystyle{plain}

\end{document}